\documentclass[11pt]{article}
\usepackage[dvipsnames]{xcolor}
\usepackage[utf8]{inputenc}
\usepackage{fullpage}
\usepackage{float}
\usepackage{amsmath}
\usepackage{amsthm}
\usepackage{amssymb}
\usepackage{xcolor}
\usepackage{todonotes}
\usepackage{hyperref}
\usepackage[noabbrev,capitalize]{cleveref}
\usepackage[linesnumbered]{algorithm2e}
\Crefname{algocf}{Algorithm}{Algorithms}
\usepackage{bm}
\usepackage{graphicx}
\usepackage{enumitem}
\usepackage{subcaption}
\usepackage{tikz}
\usetikzlibrary{matrix,decorations.pathreplacing,calc,fit,backgrounds,fit}
\usepackage{complexity}
\usepackage{mathpazo,bbm}

\AtBeginDocument{
\DeclareSymbolFont{AMSb}{U}{msb}{m}{n}
    \DeclareSymbolFontAlphabet{\mathbb}{AMSb}}

\newcommand{\ignore}[1]{\relax}

\newlist{alphaenumerate}{enumerate}{1}
\setlist[alphaenumerate]{label={\alph*.},ref={\alph*}}

\newcommand\myshade{85}
\colorlet{mylinkcolor}{violet}
\colorlet{mycitecolor}{YellowOrange}
\colorlet{myurlcolor}{Aquamarine}

\hypersetup{
  linkcolor  = mylinkcolor!\myshade!black,
  citecolor  = mycitecolor!\myshade!black,
  urlcolor   = myurlcolor!\myshade!black,
  colorlinks = true,
}

\newtheorem{theorem}{Theorem}[section]

\newtheorem{conjecture}[theorem]{Conjecture}
\newtheorem{lemma}[theorem]{Lemma}

\newtheorem{definition}[theorem]{Definition}
\newtheorem{fact}[theorem]{Fact}

\newtheorem{proposition}[theorem]{Proposition}
\newtheorem{corollary}[theorem]{Corollary}
\newtheorem{observation}[theorem]{Observation}

\newcommand{\event}{E}

\newcommand{\set}[1]{\left\{#1\right\}}

\newcommand{\fix}{\textrm{fixed}}
\newcommand{\setX}{\mathcal{X}}

\newcommand{\setY}{\mathcal{Y}}

\newcommand{\setE}{\mathcal{E}}

\newcommand{\mintropy}{H_{\infty}}
\newcommand{\dinfy}{D_{\infty}}
\newcommand{\card}[1]{\left| #1 \right|}

\newcommand{\detp}{\ensuremath \Pi}

\newcommand{\indgadget}{\text{IND}}

\newcommand{\vecspan}{\textrm{span}}
\newcommand{\bits}{\set{0,1}}

\newcommand{\search}{\mathrm{Search}}
\newcommand{\height}{C}
\newcommand{\treesize}{\mathrm{tree}}
\newcommand{\size}{\mathrm{size}}
\newcommand{\dependent}{\ensuremath{D}}
\newcommand{\rowreduce}{\textbf{row-reduce}}
\newcommand{\codim}{\mathrm{codim}}
\newcommand{\OmegaTilde}{\Tilde{\Omega}}

\newcommand{\width}{\mathrm{width}}

\colorlet{pink}{red!40}
\colorlet{blue}{cyan!60}

\title{On Disperser/Lifting Properties of the Index and Inner-Product Functions\thanks{To appear at ITCS 2023}}
\author{Paul Beame\thanks{Research supported by the National Science Foundation under NSF grant CCF-2006359} \\ \small University of Washington \\ \small \href{mailto:beame@cs.washington.edu}{beame@cs.washington.edu}\and Sajin Koroth\thanks{Research supported by the Natural Sciences and Engineering Research Council of Canada (NSERC), Discovery Grant RGPIN-2022-05211} \\ \small University of Victoria\\ \small \href{mailto:skoroth@uvic.ca}{skoroth@uvic.ca}}
\date{\today}

\begin{document}

\maketitle

\begin{abstract}

\begin{sloppypar}
Query-to-communication lifting theorems, which connect the query complexity of a Boolean function to the communication complexity of an associated `lifted' function obtained by composing the function with many copies of another function known as a gadget, 
have been instrumental in resolving many open questions
in computational complexity.
A number of important complexity questions could be resolved if we could make substantial improvements in the input size required for lifting with the Index function, which is a universal gadget for lifting, from its current 
near-linear size down to polylogarithmic in the number of inputs $N$ of the original function or, ideally,
constant.
The near-linear size bound was recently shown by
Lovett, Meka, Mertz, Pitassi and Zhang~\cite{DBLP:conf/innovations/LovettMMPZ22} 
using a recent breakthrough improvement on the Sunflower Lemma to show that a
certain graph associated with an Index function of that size is a disperser.
They also stated a conjecture about the Index function that is essential for further improvements in the size required for
lifting with Index using current techniques.
In this paper we prove the following;
\end{sloppypar}
\begin{itemize}
    \item The conjecture of Lovett et al.\ is false
    when the size of the Index gadget is less than logarithmic in $N$.
    \item The same limitation applies to the Inner-Product function. 
    More precisely, the Inner-Product function, which is known
    to satisfy the disperser property at size
    $O(\log N)$, also does not have this property when
    its size is less than $\log N$.
    \item Notwithstanding the above, we prove a
    lifting theorem that applies to Index gadgets of
    any size at least 4 and yields lower bounds for a restricted class
    of communication protocols in which one of the players is limited to sending parities of its inputs.
    \item Using a modification of the same idea with improved lifting parameters we derive a strong lifting theorem from decision tree size to parity decision tree size.
    We use this, in turn, to derive a general lifting
    theorem in proof complexity from tree-resolution size to tree-like $Res(\oplus)$ refutation size, which yields many new exponential lower
    bounds on such proofs.
\end{itemize}
\end{abstract}

\thispagestyle{empty}
\newpage
\setcounter{page}{1}


\section{Introduction}

In recent years, a substantial number of long-standing problems~\cite{DBLP:journals/toc/GargGK020,DBLP:journals/siamcomp/Goos0018,LRS15conf,DBLP:journals/siamcomp/GoosP18,RPRC16conf}
have been resolved using the method of \emph{lifting}.   
Lifting results take a gadget function $g$ and show
that any function $f:\bits^N\rightarrow\bits$ that is 
hard to compute by decision trees can be modified to a new
function $F=f\circ g^N$ that is hard for a more powerful
computational model, typically that of 2-party
communication complexity, in which case
$g:X\times Y\rightarrow\bits$ and the inputs for $F$ for the
two players are
partitioned into $\mathbf{x}\in X^N$ and $\mathbf{y}\in Y^N$.

A particularly natural and important choice of gadget $g$ is the \emph{Index} gadget $\indgadget_m:[m]\times \bits^m$ given by 
$\indgadget_m(x,y)=y_x$.
The Index gadget is universal for all gadgets 
$g:X\times Y\rightarrow\bits$ where $|X|=|Y|=m$ via the simple
reduction where $y$ is replaced by the string $(g(x,y))_{x\in X}$.
Since $\indgadget_m$ has a 2-party protocol of cost $\log_2 m +1$,
the communication complexity of $f\circ \indgadget_m^N$ is at
most $O(C^{dt}(f)\log m)$, where $C^{dt}(f)$ is the decision
tree complexity of $f$.   Another important gadget $g$ is the Inner-Product function $\IP_b:\set{0,1}^b\times\set{0,1}^b\rightarrow\set{0,1}$ given by $\IP_b(x,y)=x\cdot y\bmod 2$.

An important limitation on the quality of lower bounds that can be
proven by 
lifting with a gadget 
$g:X\times Y\rightarrow \set{0,1}$ comes from the fact that the input size for $F$ 
grows by a factor of $\log_2 |X|+\log_2 |Y|$ bits from that
of $f$. 
This limits the lower bounds on the lifted function $F$
compared to the input size of $F$.   
The original lifting theorems of \cite{RazM97conf} and \cite{DBLP:journals/siamcomp/GoosP018} used Index gadgets with $m$ a large polynomial in $N$.  
Subsequently, \cite{CKLM19journ, WuYY17eccc} proved a lifting theorem for Inner-Product gadgets with $b=c\log_2 N$ for some constant $c>5$. Later,~\cite{CKLM19journ} improved $c$ to almost 2.
The first lifting theorem for randomized computation was proved by~\cite{DBLP:journals/siamcomp/GoosPW20} again for Index but for $m$ an even larger polynomial in $N$ and later again, by \cite{DBLP:journals/siamcomp/ChattopadhyayFK21}, for Inner-Product for $b$ a larger constant multiple of $\log_2 N$ than for deterministic lifting.

A key question asked in a precursor paper to these lifting theorems~\cite{DBLP:journals/siamcomp/GoosLMWZ16} is whether lifting is possible with a sub-logarithmic or even constant-size Inner-Product gadget.
Smaller gadgets imply sharper lifting results and
more general classes of functions for which lifting
may be used to prove lower bounds.  
Proving such lifting theorems would imply breakthrough results in other areas. 
For example, proving lifting theorems with constant-size gadgets would give us a near-complete understanding of communication complexity of lifted search problems and would imply breakthrough results in associated areas like proof complexity and circuit complexity\footnote{In this paper we focus on the setting of query-to-communication lifting where the query complexity of $f$ is lifted to the communication complexity of $F$ using the gadget $g$. 
There are other lifting theorems (see~\cite{Sherstov11journ}) which lift analytical parameters of the function $f$ to the communication complexity of $F$. 
In these settings, lifting theorems with constant-size gadgets are known~\cite{Sherstov11journ} 
but for many interesting applications of lifting, there is a significant gap between analytical parameters of $f$ like approximate-degree (used in~\cite{Sherstov11journ}) and the query complexity of $f$. 
Thus, such lifting theorems with constant-size gadgets are not enough to give the results alluded to above.}.
Even improving the gadget size for Index to poly-logarithmic in $N$ would improve the best known monotone circuit size lower bounds~\cite{HR00conf} from $2^{\OmegaTilde(n^{1/3})}$ to $2^{\OmegaTilde(n)}$.
In the dream range of constant size, by the universality of Index, if there is any lifting theorem for any constant-size gadget,
there would be one for constant-size Index
gadgets.

Recent work by Lovett, Meka, Mertz, Pitassi and Zhang~\cite{DBLP:conf/innovations/LovettMMPZ22} used a new bound for the Sunflower Lemma~\cite{ALWZ:sunflower-journal} to improve the size of the Index gadget
that can
be used in deterministic lifting results to $O(N\log N)$.
They also identified a conjecture
regarding entropy deficiency and the disperser property of the Index gadget that is essential for further reductions in gadget size using current techniques.

Before stating the conjecture of Lovett, Meka, Mertz, Pitassi and Zhang~\cite{DBLP:conf/innovations/LovettMMPZ22} we give an
outline of the meta-technique for proving query to communication lifting theorems, known as the simulation theorem framework.

\subsection*{The lifting paradigm}   
The general paradigm for proving a query-to-communication lifting theorem is a step-by-step simulation argument that begins with a
communication protocol $\detp$ for $f\circ g^N$ on inputs in $X^N\times Y^N$ and derives a decision tree $T$ computing $f$ on inputs $z\in \set{0,1}^N$.  

Beginning at the root of $\detp$ and with $T$ a single root node, the simulation proceeds
to follow a path in $\detp$ maintaining sets of
inputs $\setX\subseteq X^N$ and $\setY\subseteq Y^N$ 
consistent with the current  node $u$ in protocol $\detp$.  (The exact procedure for choosing the path and  the sets $\setX$ and $\setY$ varies.)

At any point in time when $\setX\times \setY$ has revealed too much about the value of $z_i=g(x_i,y_i)$
for some $i$, the simulation at the current leaf 
node $v$ of $T$ queries $z_i$ and adds the children
$v'$ and $v''$ to $T$, one for each outcome.
The simulation then splits into cases depending on whether the $0$ or $1$ out-edge from $v$ is being followed.   
There may be multiple $i$ for which this may need
to be done at the same time.

The simulations maintain several invariants
at each corresponding pair of nodes $u$ in $\detp$ and $v$ in $T$ that occur in this simulation.
In particular, if  $\setX$ and $\setY$ are associated with this pair of nodes and
$I\subseteq [N]$ is the set
of input indices queried on the path in $T$ to $v$
and $z_I$ is the assignment that takes $T$ to the node $v$
then we require
\begin{align*}
    g^I(\setX,\setY)&=z_I&\mbox{(Consistency)}\\
    g^{[N]\setminus I}(\setX,\setY)&=\set{0,1}^{[N]\setminus I}.\qquad&\mbox{(Disperser/Extensibility)}
    \end{align*}
The consistency property is obviously required for
correctness.
The disperser property is required
because the simulation cannot
predict what query indices will be needed for $T$ 
in the future.
Overall, in order to yield a good complexity bound,
the argument also has to bound the length of the path to $v$ in
$T$, which is the size of the set $I$, as a function
of the length of the path from the root to $u$ in $\detp$.

In order to maintain these properties, the simulations also maintain some ``nice'' structure 
on  the sets $\setX$ and $\setY$. 
The most common notions of nice structure are
small \emph{entropy deficiency} of
the induced distributions on the unqueried coordinates $[N]\setminus I$ or high \emph{min-entropy rate} (equivalently
the \emph{block min-entropy})\footnote{Many existing results for the Index gadget use bounds on the min-entropy rate on the $X^N$ side and entropy deficiency on the $Y^N$ side.}. 
The min-entropy rate is the minimum ratio of the min-entropy of the induced distributions on any subset of unqueried blocks compared to the maximum possible entropy on those blocks.

There are other "nice" properties that were used in the past. For example, one of the first lifting theorems by Raz and McKenzie~\cite{RazM97conf} used a combinatorial notion of niceness defined as average-degrees in a layered graph corresponding to $\setX$. This property was also used in the later reproving of the result by~\cite{DBLP:journals/siamcomp/GoosP018}, and the result on extending deterministic lifting theorems to a larger class of gadgets including Inner Product by~\cite{WuYY17eccc,CKLM19journ}\footnote{\cite{CKLM19journ} uses it slightly differently from the application of the property for the Index gadget.}. All of the results using this combinatorial property crucially depend on a transformation in layered graphs from average degree to minimum degree known as the ``thickness lemma'' to prove the disperser property of the gadget. It is a folklore result that such average-degree to min-degree transformations do not work for Index gadgets of linear size. Thus, using "average-degree" as the nice property cannot yield lifting theorems with sublinear-size Index gadgets using existing techniques.

\subsection*{The LMMPZ conjecture on entropy deficiency and disperser properties of $\indgadget$}

The conjecture of Lovett et al.~\cite{DBLP:conf/innovations/LovettMMPZ22} is a necessary condition for small entropy deficiency to be sufficient for lifting with the Index function.
To motivate the parameters of the conjecture we first note how entropy deficiency relates to the numbers of bits of communication sent in the protocol $\detp$:

In the course of following a path in  $\detp$ to a node $u$, each bit communicated may 
split the set of consistent inputs in either $X^N$
or $Y^N$ by a factor of 2, which increases the
entropy deficiency by 1.   
If good min-entropy rate is also required,
additional pruning must be done, which further increases the entropy deficiency.
Therefore, the best one can do in terms of maintaining small entropy deficiency is to maintain a bound $\Delta$ on
entropy deficiency for $\setX$ and $\setY$ that is proportional to
the number of bits sent in $\detp$.
Bounding the length of the path in the decision in terms of
the number of bits sent means that $|I|$ should not be too large as a function of $\Delta$.
This
led Lovett et al.\ to formulate the following 
conjecture on the disperser properties of
Index as a first step towards obtaining
lifting theorems for small gadget sizes:
\begin{conjecture}[{\cite[Conjecture 11]{DBLP:conf/innovations/LovettMMPZ22}}]
\label{conj:TonisConj}
There exists $c$, such that for all large enough $m$ the following holds: Let $\setX,\setY$ be distributions on $[m]^N$, ${(\bits^m)}^{N}$, respectively, each with entropy deficiency at most $\Delta$. Then $\indgadget_m^N(\setX,\setY)$ contains a sub-cube of co-dimension at most $c \Delta$.
That is, there exists $I\subseteq [N]$, $|I|\leq c\Delta$, and $\gamma \in \bits^I$ such that for all $z\in \bits^N$ with $z_I=\gamma$ we have
{ \setboolean{@fleqn}{false}
\begin{equation}
\Pr_{x\sim\setX,\ y\sim\setY} [\indgadget_m^N(x,y)=z] > 0.\label{cond:conjecture}\\[-1ex]
\end{equation}
}
\end{conjecture}

\paragraph*{Remark:}
Note that the condition \cref{cond:conjecture} is somewhat weaker than
the combination of the consistency and disperser
conditions in the above lifting paradigm.  The lifting paradigm would correspond to additionally requiring that the $(x,y)$ pair in \cref{cond:conjecture} come from some $\setX'\subseteq\setX$ and $\setY'\subseteq\setY$ such that $\indgadget_m^I(\setX',\setY')=\gamma$. 
Here, one could satisfy \cref{cond:conjecture} using pairs $(x,y)$ and $(x',y')$ with $\indgadget_m^I(x,y)=\indgadget_m^I(x',y')=\gamma$ but
$\indgadget_m^I(x,y')\ne\gamma$.
\vspace*{1ex}

The results in~\cite{DBLP:conf/innovations/LovettMMPZ22} prove the conjecture for $m=O(N \log N)$,
in fact the stronger version with separate consistency and disperser properties required for lifting; previously it was only known when $m \gg N^2$. 
Based on a related statement about $p$-biased $(\setX,\setY)$ proved in the Robust Sunflower Theorem from~\cite{ALWZ:sunflower-journal}, the authors~\cite{DBLP:conf/innovations/LovettMMPZ22} also suggest that it is hopeful to prove the conjecture when $m=\poly(\log N)$ using techniques from their work and~\cite{ALWZ:sunflower-journal}.

\subsection*{Our results}

We disprove the LMMPZ conjecture when $m$ is  
$\log_2 N-\omega(1)$, even when $\setX$ and $\setY$ are also assumed to have extremely high min-entropy rate.
In our counterexample the distribution for 
$\setX$ is uniform on $[m]^N$ and so has
full entropy and maximum possible min-entropy rate.
The distribution on $\setY$ is also uniform so we
view both $\setX$ and $\setY$ as subsets of
$[m]^N$ and $(\bits^m)^N$ respectively.

Though the parameter $\Delta$ governing the entropy  deficiency in the conjecture is universally quantified, the failure of the
conjecture occurs over a very wide range of values of $\Delta$.
In fact, when $m\le (1-\alpha)\log_2 N$, we prove
a much larger gap and show
that $|I|$ must be 
$\Omega(N^\alpha)$ independent of $\Delta$ for \cref{cond:conjecture} to hold.

\begin{theorem}
\label{thm:main}
For any $\Delta\ge 1$ and $m$ with $2^m\le N/(K\Delta)$
for $K\ge 1$ there is a set 
$\setY\subseteq (\bits^m)^N$ of entropy deficiency at most $\Delta$ and min-entropy rate at least
$1-1/m$ such that
for every $I\subseteq [N]$ with $\card{I} \leq (K-1) \Delta$,
the set $\indgadget_m^{N\setminus I}([m]^N,\setY)$ does not contain the all-0 string.
\end{theorem}

Since $\setX=[m]^N$ has no deficiency (and min-entropy
rate 1) we immediately derive the following:

\begin{corollary}
\cref{conj:TonisConj} is false when
$m\leq  \log_2 (N/\Delta)-\omega(1)$.
Moreover, for all $\Delta$, when 
$m\le (1-\alpha)\log_2 N$, for any set
$I$,
$|I|$ must be 
$\Omega(N^\alpha)$ for \cref{cond:conjecture} to hold.
\end{corollary}

Furthermore, an analogous property applies to the Inner-Product function:

\begin{theorem}
\label{thm:inner-product}
For $b\le \log_2(N/(K\Delta))$ for $K\ge 1$, there is a set 
$\setY\subset (\bits^b)^N$ of entropy deficiency at most $\Delta$ and min-entropy rate more than
$1-1/b$ for every $I\subseteq [N]$ with $\card{I} \leq (K-1) \Delta$, the set
$\IP_b^{N\setminus I}((\bits^b)^N,\setY)$
does not contain the all-1 string.
\end{theorem}

Therefore, though lifting theorems, both deterministic and randomized, have already been proven for
Inner-Product gadgets on $c\log_2 n$ bits using only properties of small entropy-deficiency and high min-entropy rate~\cite{DBLP:journals/siamcomp/ChattopadhyayFK21}, using these properties we can at best reduce the Inner-Product gadget
size in such
lifting theorems by at most a constant factor since lifting for Inner-Product gadgets with significantly fewer than $\log_2 n$ bits are impossible using those properties.

The proof idea for these theorems is quite simple and 
relies on the fact that for such small values of $m$,
it is likely that a uniformly random string $y$ will 
have many blocks $i$ where $y_i=0^m$ and hence
cannot have 1 output values in any of those coordinates.

Despite this setback, the dream of lifting theorems for constant-size gadgets remains.
Our second main result is that, though we can rule out the disperser properties of \cref{conj:TonisConj} for $\indgadget_m$ with sub-logarithmic $m$, 
there is an interesting class of protocols, one in
which Bob's messages are constrained to be parity functions of his input string $y$, in which we can prove a deterministic lifting theorem using 
$\indgadget_m$ gadget for \emph{constant size} $m$.

Since Alice is unrestricted and Bob is restricted, 
we call such protocols semi-structured protocols.    
We obtain a lifting theorem for semi-structured protocols showing that the decision tree height is asymptotically at most a $1/\log m$ fraction of the complexity of 
 the communication protocol for the lifted function with $\indgadget_m^N$.
As is typical for deterministic lifting theorems, this works
both for functions and for search problems.
For protocols in which both Alice and Bob only send parities of their inputs, we obtain an even stronger simulation that
applies to the size of the decision tree produced in terms of the number of leaves (size) of the communication protocol for the
lifted function.

In particular a modification of this idea gives us a generic theorem that lifts decision tree lower bounds of height $t$ 
or size $s$ 
for any explicit function $f$ on $n$ inputs to a corresponding lower bound for parity decision trees of height $\Omega(t)$
or size $\Omega(s)$ 
for an explicit function $f'$ on $O(n)$ inputs.

The latter also yields new lower bounds for tree-like proofs in the
$Res(\oplus)$ proof system  introduced by
Itsykson and Sokolov~\cite{DBLP:conf/mfcs/ItsyksonS14,DBLP:journals/apal/ItsyksonS20}
who
proved tight exponential lower bounds for the pigeonhole principle as well as exponential lower bounds for a very restricted kind
of lifted formula based on Tseitin formulas from~\cite{bps:kfoldtseitin-journal}.  
(This system is also known as $ResLin_2$ because of its relationship to the $ResLin$ proof system of
Raz and Tzameret~\cite{DBLP:journals/apal/RazT08}.)
Huynh and Nordstr\"{o}m~\cite{DBLP:conf/stoc/HuynhN12} gave lifting theorems for a variety of other proof systems using constant-size Index gadgets (indeed with $m=3$) but these only yield good bounds for a restricted class of formulas whose search
problems have high ``critical block sensitivity".
Here we obtain exponential lower bounds for a substantially broader class of formulas.   
In particular, for any of the vast class of $k$-CNF formulas $\varphi$ for which exponential tree-resolution lower bounds are known, we obtain lifted
$O(k)$-CNF formulas $\varphi'$ with a constant factor increase in number of variables (and a constant factor increase in number of clauses if $k$ is constant) 
requiring tree-like $Res(\oplus)$ refutations of 
exponential size (indeed at least the tree-like resolution refutation size for
$\varphi$).

\paragraph*{Related Work}
Independently of our work,
Chattopadhyay, Mande, Sanyal, and Sherif~\cite{ChattopadhyayMSS23-arxiv} have obtained closely related lifting results for parity decision tree size and the size of tree-like $Res(\oplus)$ proofs. 
Their lifting theorem works not only for lifting with Index gadgets but, more generally, for lifting with a class of gadgets that includes Inner-Product and other simple gadgets.
Their methods and ours have considerable
similarity, particularly in the use of
row-reduction as a key component.

\section{Preliminaries}

\paragraph*{Notation:} For a set of vectors (or a distribution $\setX$) on $U^N$ and $I\subseteq [N]$, we use $\setX_I$ to denote
the projection of $\setX$ onto the coordinates in $I$.
For a function $h$ on $U$, we let $h^I(\setX)$ 
denote the set of all possible vectors of outputs of $h^I$ on 
$\setX_I$; if this set is a singleton $w$ we abuse notation and simply define the value to be $w$.
\paragraph*{Information theory:} 
We use several definitions for forms of entropy:
\begin{definition}The \emph{entropy deficiency} (sometimes simply \emph{deficiency}) of a distribution $\mathcal{X}$ on 
a universe $U$, $\dinfy(\setX)$, is 
$\log_2 |U|-H_2(\mathcal{X})$. 
For a subset $\mathcal{V}\subseteq U$, the deficiency of $\mathcal{V}$ is that of the uniform
distribution on $\mathcal{V}$.
In particular, for example, the deficiency of a set of inputs $\setY \subseteq (\bits^m)^N$ of Bob satisfies
$$2^{-\dinfy(\setY)}=\frac{ \card{\setY} } { 2^{m \cdot N} }.$$
\end{definition}

\begin{definition}
The \emph{min-entropy} of a distribution 
$\mathcal{X}$ on $U$,
$H_\infty(\mathcal{X})$, is 
$$\min_{x\in U}\log_2(1/\Pr_{\mathcal{X}}(x)).$$
For a distribution $\setX$ on a set $U^N$, the 
\emph{min-entropy rate} of $\setX$ is the maximum $\tau$ such that for every $J\subseteq [N]$,
$H_\infty(\setX_J) \geq \tau |J|\log_2 |U|$ or, 
equivalently, such that for all $\alpha_J\in U^J$,
$$\Pr_{x\sim \setX}[x_J=\alpha_J]\le |U|^{-\tau\cdot |J|}.$$
\end{definition}

We use the following fundamental fact about the effect of conditioning on the min-entropy of a random variable.

\begin{fact}\label{fact:conditioningMinEntropy}
Let $\setX$ be a random variable and let $\event$ be an event. Then $\mintropy(\setX \mid \event) \geq \mintropy(\setX) - \log_2 (1/\Pr[\event])$.
\end{fact}

\paragraph*{Probability:}
We will use the following nice bound on the median of
any binomial distribution.

\begin{proposition}[\cite{kb:binomial-median}]
\label{prop:binomomial-median}
The median of a binomial distribution $B(n,p)$ lies between
$\lfloor np\rfloor$ and $\lceil np\rceil$.
\end{proposition}


\paragraph*{Linear algebra:}
We use row-reduced form to represent matrices in our simulation theorems.
\begin{definition}[Row-reduced matrices]
\label{def:row-reduced}
    A matrix $M$ with $r$ rows is said to be \emph{row reduced} if it contains an $r\times r$ identity
    submatrix.
\end{definition}




\paragraph*{Parity decision trees
and $Res(\oplus)$ refutations:}
A \emph{parity decision tree} over a set of Boolean variables
$Z$ defining a space $\set{0,1}^Z$ of Boolean vectors is a rooted binary tree in which each internal node is labeled by a parity of variables from $Z$ with out-edges labeled 0 and 1 respectively.
Each leaf is labeled by an output value.  Such a decision tree computes a function on $\set{0,1}^Z$.

For a function or relation $f$ on Boolean inputs, let $\height^{dt}(f)$ and $\size^{dt}(f)$ be the
minimal height and size, respectively of any decision tree computing $f$ and let $\height^{\oplus dt}(f)$ and
$\size^{\oplus(dt)}(f)$ be the corresponding
measures for parity decision trees.

Before defining the proof system $Res(\oplus)$ and its relationship to parity decision trees, we first review the resolution proof system and its relationship to ordinary decision trees.

A \emph{resolution} ($Res$) refutation of an unsatisfiable CNF formula $\varphi$ on variables $Z$
is a sequence of clauses ending in the empty clause
$\bot$ in which each clause is either
a clause of $\varphi$, or follows from two prior
clauses using the inference rule
$$\frac{A\lor z_i,\ B\lor \overline {z_i}}{A\lor B}$$
for some variable $z_i\in Z$; we say that this step
\emph{resolves} on variable $z_i$.

This sequence yields a
directed acyclic graph (dag) of in-degree 2 each of whose nodes is labeled by a clause on the variables in $Z$, with sources labeled by clauses of $\varphi$ and sink
labeled by $\bot$.
The resolution refutation is \emph{tree-like} if this
associated dag is a tree and the associated refutation is
called a \emph{tree-resolution} refutation of $\varphi$.
For an unsatisfiable CNF formula $\varphi$, let $\treesize_{Res}(\varphi)$ be the minimum size of
any tree-resolution refutation of $\varphi$.
Further, let $\width_{Res}(\varphi)$ denote the minimum over all
resolution refutations of $\varphi$ of the length of the longest clause
in the refutation.
We have the following result of Ben-Sasson and Wigderson~\cite{DBLP:journals/jacm/Ben-SassonW01}.

\begin{proposition}
\label{prop:bsw}
For any CNF formula $\varphi$ with clause-size at most $k$,
$\treesize_{Res}(\varphi)\ge 2^{\width_{Res}(\varphi)-k}$.
\end{proposition}

Every unsatisfiable CNF formula $\varphi$ yields an
associated total search problem $\search_\varphi$ which
takes as input an assignment $z\in\set{0,1}^Z$ and produces the name of a clause that
is falsified by $z$.

\begin{proposition}
\label{prop:dpll}
Given any unsatisfiable formula $\varphi$,
minimal tree-resolution refutations of $\varphi$ and
decision trees solving $\search_\varphi$ are isomorphic.
This isomorphism identifies nodes that resolve on a 
variable $z_i$ with those that branch on variable $z_i$.
\end{proposition}

Each node in a decision tree naturally corresponds to the sub-cube of the input set $\set{0,1}^Z$ given by the constraints on the path to the node from the root.   
Similarly, any clause $C$ can be identified with a sub-cube of $\set{0,1}^Z$ consisting of the set of all inputs falsified by $C$.   
We say that a clause $C$ is a \emph{weakening} of a clause $A$ iff there is some clause $B$ such that $C=A\lor B$; alternatively this is equivalent to saying
that the sub-cube corresponding to $C$ is contained in
the sub-cube corresponding $A$.
With this correspondence, the isomorphism in the above
proposition means that the sub-cube for each node in the
decision tree corresponds to a weakening of the clause at the isomorphic node in the tree-resolution refutation.

A \emph{$Res(\oplus)$ refutation} of an unsatisfiable CNF formula $\varphi$ in variables $Z$ is a sequence of affine subspaces of $\mathbb{F}_2^Z$ ending in
the subspace $\mathbb{F}_2^Z$ such that each subspace in the list is either
the sub-cube corresponding to a clause of $\varphi$ or follows from two prior
subspaces $A$ and $B$ via the inference rule:
$$\frac{A,\ B}{C}\qquad \mbox{if } C\subseteq A\cup B.$$

(Alternatively, one can replace each subspace by an expression for its dual, namely a disjunction of parity equations over $\mathbb{F}_2$, each of which is the negation of one of a set of linear equations defining the affine subspace.
Therefore, each line is a clause over parities and the subspace is the set
of inputs that falsifies it.
In this way, the single inference rule generalizes the weakening rule of resolution and generalizes resolving on a single
literal to resolving on a parity of variables; it is not hard to see that the subspace $C$ must be contained
entirely in $A$ or in $B$ unless there is a unique
linear constraint defining $A$ whose negation is a defining equation for $B$. 
This representation is not unique since there may be many different choices
of defining equations for an affine subspace, but we assume that the 
proof is representation independent.
As noted by Isykson and Sokolov~\cite{DBLP:conf/mfcs/ItsyksonS14,DBLP:journals/apal/ItsyksonS20},
the semantic view of $Res(\oplus)$ refutations we have presented is a standard proof system in the sense of Cook and Reckhow~\cite{DBLP:journals/jsyml/CookR79} since inference is explicitly verifiable in polynomial time.)

As with resolution refutations, we can define the dag of indegree (at most) 2 associated with a $Res(\oplus)$
refutation and let \emph{tree-$Res(\oplus)$} be the proof system consisting
of $Res(\oplus)$ refutations whose associated dag is a tree.
For an unsatisfiable CNF formula $\varphi$, let $\treesize_{Res(\oplus)}(\varphi)$ be the minimum size of
any \emph{tree-$Res(\oplus)$} refutation of $\varphi$.
Using the analogous ideas to the isomorphism between
decision trees and minimal tree-resolution proofs, Itsykson and
Sokolov proved the following correspondence:

\begin{proposition}\cite{DBLP:conf/mfcs/ItsyksonS14,DBLP:journals/apal/ItsyksonS20}
\label{prop:reslin2}
Given any unsatisfiable formula $\varphi$,
minimal tree-$Res(\oplus)$ refutations of $\varphi$ and
parity decision trees solving $\search_\varphi$ are isomorphic.
This isomorphism identifies nodes that resolve on a 
parity function $\oplus_S(z)$ with those that branch on $\oplus_S(z)$.   
\end{proposition}

\paragraph*{Lifting CNF formulas with $\indgadget_m$:}
There are a few options for how to do this.   
Since we will use this when $m=2^\ell$ is a constant, we choose a simple option that has $N\ell$ Boolean variables $x_{i,j'}$ for $j'\in \set{0,\ldots, \ell-1}$ and $i\in [N]$ and $Nm$ Boolean variables $y_{i,j}$ for $j\in\set{0,\ldots,m-1}$ and $i\in [N]$.
As usual, the interpretation we have for $\indgadget_m^N$ is that $z_i=y_{i,x_{i,\ell-1}...x_{i,0}}$.

Given a $k$-CNF formula $\varphi$ in $N$ variables $Z=\set{z_1,\ldots,z_N}$,
we define an $(\ell+1)k$-CNF formula $\varphi\circ \indgadget_m^N$ on
the $x_{i,j}$ and $y_{i,k}$ variables as follows:
Each clause $(z_{i_1}^{b_1}\lor \cdots z_{i_k}^{b_k})$ of
$\varphi$ is replaced by $m^k$ clauses of length
$(\ell+1) k$, one for each
tuple $(j^{(1)},\ldots,j^{(k)})$, which expresses the
statement that if the $\ell$ bits $x_{i_1,*}$ encode
value $j_1$, those of $x_{i_2,*}$ encode $j_2,...$ and
those of $x_{i_k,*}$ encode $j_k$, then 
$y_{i_1,j_1}^{b_1}\lor \cdots y_{i_k,j_k}^{b_k}$ must be
true.
With this definition, each clause of $\varphi\circ\indgadget_m^N$ 
corresponds to a unique clause of $\varphi$.   
Moreover, a falsified clause of $\varphi\circ\indgadget_m^N$ on one of its input vectors
yields a falsifying assignment to the corresponding
clause of $\varphi$ under the vector of $z$ values
given by $\indgadget_m^N$.

\section{On the insufficiency of low deficiency and
high min-entropy rate}

In this section we prove \cref{thm:main,thm:inner-product}.

Let $K$ be any function of $N$ with $K\ge 1$
for all $N$ 
and assume that $2^m\le N/(K\Delta)$.
We will construct a specific distribution $\setY$, with deficiency bounded by $\Delta$ and min-entropy rate at
least $1-1/m$, such that for any $I\subseteq [N]$ with $|I|\leq (K-1)\Delta$, $\indgadget_m^{[N]\setminus I}([m]^N,\setY)$ does not contain the all-1 string. 

We establish this simply by showing that for all $(x,y)\in [m]^N \times \setY$, $\indgadget_m^N(x,y)$ has Hamming weight more than $(K-1)\Delta$. 
Thus any projection of $\indgadget(\setX,\setY)$ onto $N-|I|\ge N-(K-1)\Delta$ coordinates will contain at least one $1$, and would therefore miss the all-0 string. 

To this end, we will construct $\setY$ as the uniform distribution on a subset $S\subseteq (\bits^m)^N$ with the following properties :
\begin{enumerate}
    \item\label{prop:AtLeastKOnes} Every $y\in S$ has at least $k > (K-1)\Delta$ blocks that are equal to $1^m$.
    \item\label{prop:LowDeficiency} $|S|\ge 2^{mN-\Delta}$
    \item\label{prop:HighMinEntropyRate} The min-entropy
    on any subset of $b\le N$ blocks is at least $(m-1)b$.
\end{enumerate}
The second property ensures that the deficiency of $\setY$ is at most $\Delta$ and the first property guarantees that every
output has at least $k$ 1-bits no matter what the input $x$ is. 
The third property is simply that the min-entropy rate
is at least $(1-1/m)$.

The essence of the proof idea applies in the case that $\Delta=1$; the general case is a simple extension
of that special case:

 \paragraph*{Counterexample when $\Delta=1$:} 
 We derive a counterexample in this case
 by choosing $k=\lfloor K\rfloor>K-1$ and setting $S$ to be the set of all inputs in $(\bits^m)^N$ that have at least $k$ blocks of the form $1^m$; i.e., all-1 blocks.
 
The key observation that makes this work is the 
following:
 
 \begin{observation}
 \label{all-1}
For $y$ chosen uniformly at random from $(\bits^m)^N$, the number of all-1 blocks in $y$ is distributed according to the binomial distribution $B(N,1/2^m)$.
\end{observation}

In  particular, since $\Delta=1$ this means that the expected number of all-1
blocks in $y$ chosen uniformly from $(\bits^m)^N$ is 
$N/2^m\ge K$.
By applying a bound on the median of binomial distributions, we obtain the following:

\begin{lemma}
\label{all-1-tail}
For $2^m\le N/K$, at least 
$1/2$ of all strings in $(\bits^m)^N$ have more than
$K-1$ all-1 blocks.
\end{lemma}

\begin{sloppypar}
\begin{proof}
By~\cref{all-1}, the number of all-1 blocks is given by the binomial distribution $B(N,1/2^m)$. 
By \cref{prop:binomomial-median}, the median of this
distribution is at least $\lfloor N/2^m\rfloor\ge \lfloor K \rfloor > K-1$.
The claim follows since
the binomial is integer-valued.
\end{proof}
\end{sloppypar}

\cref{all-1-tail} shows that $\setY$ is a uniform
distribution with deficiency
at most 1, which means that the projection on any $b$
blocks has min-entropy at least $mb-1$ and hence $\setY$
has min-entropy rate at least $1-1/m$.

\paragraph*{Counterexample for $\Delta>1$ but $o(N)$:}
For this we assume without loss of generality that $\Delta$ is
an integer and $N=N'\Delta$ for some integer $N'$ since rounding can add only $\Delta-1$ extra coordinates. 

Since $2^m\le N/(K\Delta)$, we have
$2^m\le N'/K$.
This means that we can use the counterexample distribution
$\setY'$ for the
case $\Delta=1$ on $N'$ coordinates.
We define $\setY$ to the direct 
product of $\Delta$ independent copies of the distribution $\setY'$ on disjoint coordinates.

By construction, $\setY$ has
deficiency at most $\Delta$; it also has min-entropy
rate at least that of each $\setY'$ which is $1-1/m$
since it is a product over disjoint coordinates.
Also by construction, every $y$ in the support of $\setY$ has
more than $(K-1)\Delta$ all-1 blocks which means that no
set $I$ of at most $(K-1)\Delta$ coordinates cannot
cover all of the all-1 blocks of $y$, which would be necessary to
have
the all-0 string in
$\indgadget_m^{[N]\setminus|I|}([m]^{N},\setY)$.

\paragraph*{Extending the counterexamples to other gadgets}
\label{subsec:extendingToOtherGadgets}
We note that the above result also disproves a similar conjecture for any 2-party gadget $g:X\times Y\rightarrow \bits$
whose communication matrix has a row or column that has a
constant value.
For example, the Inner-Product gadget on $\IP_b:\bits^b\times
\bits^b$ has
this property for the row or column of its communication matrix
indexed by $0^b$. 
That is, $\langle x,0^b\rangle=0$ for all
$x\in \bits^b$. 
It is easy to see that the above analysis works
to disprove the analogous conjecture for Inner-Product under the
same conditions, though
in this case, the output vector that would be missed is
the all-1 vector.
\section{Lifting theorem for semi-structured protocols}

Since \cref{conj:TonisConj} is false for $m = (1- o(1) )\log N$, without modification, existing techniques cannot reduce the gadget size below this threshold. 
However, this does not rule out other approaches to proving lifting theorems with very small gadgets, even for constant-size ones. 
Are such theorems 
with constant-sized gadgets  possible at all?   

As a first step towards answering this question, in this section we prove a (deterministic) lifting theorem with constant-size Index gadgets for a restricted family of communication protocols.

The restricted family that we consider are deterministic communication protocols in which Alice is unrestricted, but Bob is only allowed to communicate parities of his input bits. 
Since Alice is unrestricted and Bob is restricted, 
we call such protocols \emph{semi-structured} or
$(*,\oplus)$-protocols.

For a Boolean function $F: X \times Y \to \bits$,
we use $C^{*,\oplus}(F)$ to denote the
deterministic communication complexity of $F$ by 
such protocols.
Similarly, for a relation (search problem) $R\subseteq X\times Y\times W$
we use $C^{*,\oplus}(R)$ to denote the deterministic
communication complexity of such protocols solving $R$,
that is, when Alice receives $x\in X$ and Bob receives
$y\in Z$, the protocol outputs some $w\in W$ with
$(x,y,w)\in R$ or outputs $\bot$ if no such $w$ exists.

\begin{theorem}[Lifting theorem for semi-structured protocols]
\label{thm:liftingForSemiStructured}
Let $m \ge 4$ be an integer. 
For every $f: \bits^N \to \bits$, 
$$C^{*,\oplus}(f\circ \indgadget_m^N)\ge \frac{1}{2} C^{dt}(f) \log_2 m.$$
Furthermore for every $R\subseteq Z^N\times W$,
$$C^{*,\oplus}(R\circ \indgadget_m^N)\ge \frac{1}{2} C^{dt}(R)\log_2 m.$$
\end{theorem}

Let $\detp$ be a $(*,\oplus)$-protocol for $R\circ \indgadget_m^N$ of complexity 
$C^{*,\oplus}(R)$.
Without loss of generality we can assume that $C^{*,\oplus}(R)\le  N (\log_2 m +1)$.
Following the lifting theorem paradigm we prove \cref{thm:liftingForSemiStructured} by
showing how to produce a decision tree $T$ for $R$
of height at most that of $\detp$ by simulating
$\detp$.

\subsection{High level overview and invariants}
Since Bob is only allowed to communicate parity equations, we will follow the lifting
paradigm and maintain the set $\setY$ as an affine subspace over $\mathbb{F}_2^{[N]\times[m]}$.  
We will also maintain the property that  the codimension $d$ of $\setY$ is at most the total number of bits communicated during the protocol.

At every point in our simulation we will maintain
a set $\dependent_\setY\subseteq [N]\times [m]$
of \emph{dependent} coordinates. 
All other coordinates in $[N]\times [m]$ will be
\emph{free}.
We maintain $\setY$ of codimension $d$ as the 
set of a solutions of a system $\setE$ of $d$ affine equations over $\mathbb{F}_2$ in
\emph{row-reduced} form
$$\mathbf{M}\cdot y=\mathbf{b},$$
such that $\mathbf{M}$ (up to permutation of rows) is a $d\times d$ identity submatrix on the 
columns $\dependent_\setY$ that we have designated
as dependent coordinates.
It is immediate from this set-up that we have
the properties:
\begin{description}
\item[(A)] $|\dependent_\setY|=\codim(\setY)$.
\item[(B)] For every
total assignment to the free coordinates, there is a unique assignment to $\dependent_\setY$ that extends it to an element of $\setY$.
\end{description}

Bob's communication of some parity function
$\oplus_{(i,j)\in S}\  y_{i,j}$ of his
input can add at most one new affine equation to $\setE$, depending on whether or not the value of that parity function is already fixed on $\setY$.
There is no change to $\setE$ if and only if the parity is in
$\vecspan(\setE)$, the $\mathbb{F}_2$
span of the parities defining $\setE$.

As in the general lifting paradigm, we maintain a set
$I\subseteq [N]$ of \emph{fixed} indices\footnote{As in the discussion so far, to
keep notions separate we will use the term ``indices'' to refer to elements of $[N]$ and ``coordinates'' to refer to elements of $[N]\times[m]$.} on which there is a single fixed output for 
$\indgadget_m^I(\setX,\setY)$.
In addition to the above, we also maintain the following
invariants.
\begin{description}
\item[(C)] For all indices that are not fixed, elements of $\setX$ only point to free coordinates;  that is, for every $i\in [N]\setminus I$, and every $x\in \setX$ we have $(i,x_i)\notin \dependent_\setY$.
\item[(D)] $\setX$ has min-entropy rate at least
$\tau=1/2$ on $[N]\setminus I$.
\end{description}


Maintaining (C) is quite easy: Whenever we identify a new dependent coordinate $(i,j) \in [N]\times[m]$ in $\setY$, we simply update $\setX$ by removing all $x\in \setX$ with  $x_i=j$. 

We maintain (D) using the methods of G\"o\"os, 
Pitassi, and Watson (GPW) to restore the 
min-entropy rate whenever it falls too low.
Since the affine structure of $\setY$ allows for a  precise definition of dependent coordinates, our lifting theorem differs from the GPW-style lifting theorems in its parameters and philosophy\footnote{GPW never makes queries based on Bob's communication, but we do!}. These differences allow us to overcome the dependence between $m$ and $N$ in such theorems. 

\subsection{The Simulation Algorithm}

As we discussed in our high-level overview, we
maintain $\setY$ as an affine subspace of $(\bits^m)^N$ defined
by a set of linearly independent equations $\setE$
in a row-reduced form that contain an identity matrix on the set
$\dependent_\setY$ of dependent coordinates.
We need to be able to update this as new affine equations are added, either because of communication by Bob or because we have added some
$i$ to the set $I$ of fixed indices.   

For this, we define a helper function
$\rowreduce(\setE,\dependent_\setY,e)$ that takes
as input
\begin{itemize}
    \item a set of row-reduced equations $\setE$,
    \item a set $\dependent_\setY$
of coordinates for its dependent variables, and
\item a new affine equation $e$ linearly independent of
$\mathcal{E}$,
\end{itemize}
and uses Gaussian elimination to return a pair $(\setE',(i,j))$ where
$\setE'$ is equivalent to $\setE\cup\set{e}$, 
$(i,j)\notin \dependent_\setY$ and $\setE'$ is
row-reduced, as witnessed by the columns of $\dependent_\setY\cup\set{(i,j)}$.
We note that when a new equation is introduced 
while adding $i$ to the set of fixed
indices, the new equation $e$ will be of a particularly
simple form, namely $y_{i,j}=b$ where
$(i,j)\notin \dependent_\setY$ by our maintenance of invariant (C), in which case 
the new dependent coordinate returned will be
$(i,j)$.

\begin{sloppypar}

We follow a variant of the simulation algorithm of GPW\cite{DBLP:journals/siamcomp/GoosP018}, with a few modifications to identify and exclude dependent coordinates in $\setY$. 
The algorithm uses a sub-routine \textbf{RestoreMinEntropyRateAndQuery} (see  Algorithm~\ref{alg:restoreX}). This is essentially density restoration from GPW and makes sure that $\setX$ has min-entropy rate at least $\tau=1/2$ by adding
fixed indices to $I$, adding queries to the decision tree and fixing some coordinate $y_{i,\alpha_i}$ to the query answer for $z_u$. 
Note that GPW style lifting only uses this procedure to restore min-entropy rate at a node where Alice speaks, but   
we also may need this when Bob speaks because we reduce $\setX$ by
removing pointers to dependent coordinates.
Another difference in our version of density-restoration is that we only chose the first part in the partition (as we are doing deterministic lifting opposed to the randomized lifting in \cite{DBLP:journals/siamcomp/GoosPW20}).
\end{sloppypar}

\begin{algorithm}[tbp]
\DontPrintSemicolon

\caption{Simulation algorithm}
\label{algo:simulationAlgo}
\TitleOfAlgo{$\textbf{Query}_\detp(z)$}

\KwData{$z\in \set{0,1}^N$, $\setX=[m]^N$, $\setY= ( \bits )^N$, $\setE=\varnothing$, $\dependent_\setY=\varnothing$, $I=\varnothing$, $\rho=*^n$, protocol $\detp$ for $R\circ \indgadget_m^N$, $v$=root of $\detp$. }
\KwResult{element of $R(z)$.}

\While{$v$ is not a leaf}{
  Let $v_0,v_1$ be the children of $v$ following communication $0$ and $1$ respectively\;
  \If{Bob speaks at $v$} 
   {Let the parity function at $v$ be $\oplus_{(i,j)\in S}\ y_{i,j}$\;
   \eIf{$\oplus_{(i,j)\in S}y_{i,j} \in \vecspan(\setE)$}{
    $v\leftarrow v_b$ for the unique $b$ such that 
    $\oplus_{(i,j)\in S}\ y_{i,j}=b$ for all $y\in \setY$\;
    }{\tcp{Half of $\setY$ goes to $v_0$, half to $v_1$; we choose the smaller subtree.}
    Choose $b\in \bits$ with subtree rooted at $v_b$ no larger than one rooted at $v_{1-b}$.\;
    $\setY \leftarrow \set{y\in \setY \mid \oplus_{(i,j)\in S}\ y_{i,j}=b}$\;
    $\left(\,\setE,(i^*,j^*)\,\right)\leftarrow \rowreduce(\,\setE,\ \dependent_\setY, \  \oplus_{(i,j)\in S}\ y_{i,j}=b\,)$\;
    \textbf{add} $(i^*,j^*)$ to $\dependent_\setY$\label{step:AddDependentCoordinate} \;
       $\setX \leftarrow \setX \mid_{x_{i^*}\neq j^*}$ \tcp*[r]{Note: if $i^*\in I$ then $x_{i^*}\ne j^*$ already.}\label{step:updateXAvoidDependent}
     $v\leftarrow v_b$
     \; 
      }
  }
  \If{
  Alice speaks at $v$ and partitions $\setX$ into $\setX^0 \cup \setX^1$}{
  Let $b\in \bits$ be such that $\card{\setX^b}\geq \frac{1}{2} \cdot \card{\setX}$\;
  $\setX \leftarrow \setX^b$\label{step:chooseBiggerX}\; 
  $v\leftarrow v_b$\; 
  }
  \If{$\textbf{min-entropy-rate}(\setX) < \tau=1/2$}{
      \textbf{RestoreMinEntropyRateAndQuery}($\setX,z$)\;
    } 
}
\Return{label of $v$}\;

\end{algorithm}

\begin{algorithm}[tbp]
\DontPrintSemicolon
\caption{Procedure \textbf{RestoreMinEntropyRateAndQuery}}
\label{alg:restoreX}
\TitleOfAlgo{\textbf{RestoreMinEntropyRateAndQuery}($\setX,z$)}

\KwData{$\setX \subseteq [m]^N$, $z\in \{0,1\}^N$.}
\KwResult{Updates $\setX$ to restore min-entropy rate to $\tau=1/2$ by fixing coordinates via queries to $z$.}
Let $I'\subseteq [N]\setminus I$  be a maximal set on
on which $\setX$ has min-entropy rate $<\tau=1/2$\;
Let $\alpha_{I'}\in [m]^{I'}$ be such that $\Pr_{x\in \setX}[ x_{I'}=\alpha_{I'}]>m^{-\tau |I'|}$ \;
$\setX \leftarrow \set{x\in \setX\mid x_{I'}=\alpha_{I'}}$ \label{step:updateXrestoredensity}\;
\textbf{Query} all coordinates in $I'$ and let $z_{I'}$ be the query answers\;
$\setY \leftarrow \set{y\in \setY\mid  y_{(I',\alpha_{I'})}=z_{I'} }$\;
\ForEach{$i\in I'$}{   
    $\rho(i) \leftarrow z_{i}$\ \;
    $\left(\,\setE,(i^*,j^*)\,\right)\leftarrow \rowreduce(\,\setE,\ \dependent_\setY, \  y_{i,\alpha_i}=\rho(i)\,)$\tcp*[f]{Note: $i^*=i$, $j^*=\alpha_i$}\;
    \textbf{add} $(i,\alpha_i)$ to $\dependent_\setY$
     \label{step:updateRhoDueToXQuery}\;
}
$I\leftarrow I\cup I'$\;

\end{algorithm}

\subsection{Analysis of the simulation algorithm}

We first argue that $\setX$ and $\setY$ are never empty during the run of our simulation algorithm. Thus, when the algorithm reaches a leaf node of $\detp$ we can output a correct answer.
To do so we observe the following invariants on our simulation algorithm (\autoref{algo:simulationAlgo}).

\begin{lemma}[Invariants of the Simulation Algorithm (\autoref{algo:simulationAlgo})]\label{lemma:invariantsSimulationAlgo}
At the beginning of every iteration of the while loop in Algorithm~\ref{algo:simulationAlgo}, the following properties hold:
\begin{alphaenumerate}
    \item \label{prop:path} $\rho$ defines the path in the decision tree $T$ that is the outcome of the queries and $\fix(\rho)=I$.
    \item \label{prop:SetYequations} $\setY$ is the set of inputs satisfying $\setE$ which is row-reduced on $\dependent_\setY$.
    \item\label{prop:SetXAvoidsDependent} For any $x\in \setX$, and any $i\notin I$, $(i,x_i)$ is a free coordinate of $\setY$; that is $(i,x_i)\notin \dependent_\setY$.
    \item \label{prop:SetYunique} For every total assignment to the free coordinates $([N]\times [m])\setminus \dependent_\setY$, there is a unique assignment to $\dependent_\setY$ that extends it to an element of $\setY$. \item\label{prop:InvariantSetX} $\setX$ has min-entropy rate at least $\tau$ on
    $[N]\setminus I$
\end{alphaenumerate}
\end{lemma}

\begin{proof}
All but the last of the conditions of
 \cref{lemma:invariantsSimulationAlgo} easily can be seen to
 hold by inspection of \cref{algo:simulationAlgo,alg:restoreX}.
 The last follows by the argument of GPW and follows from the maximality of the set $I'$ in \cref{alg:restoreX}.
It is easy to see that if $\setX$ has non-zero min-entropy rate, then $\setX$ is non-empty. 
Moreover, since every element of $\setX$ points only to free
coordinates in blocks outside of $I$ and the min-entropy on each block is large, $\setY$ must have many free coordinates outside
of $I$.
\end{proof}

We bound the number of queries $|I|$ by using a potential function equal to the deficiency of $\setX_{[N]\setminus I}$.
Let $A$ be the number of bits spoken by Alice, and $B$
be the number of bits spoken by Bob in $\detp$.
We analyze the change in $\dinfy(\setX_{[N]\setminus I})$ due to updates of $\setX$ and $I$: 

\begin{itemize}
    \item  \autoref{step:updateXAvoidDependent} in Algorithm~\ref{algo:simulationAlgo}: removing $x_{i^*}=j^*$ from $\setX$ for newly dependent  $(i^*,j^*)$:\\ 
By \cref{lemma:invariantsSimulationAlgo}(\ref{prop:InvariantSetX}), we know that $\setX$ has min-entropy rate at least $\tau=1/2$. 
    Thus, $\Pr_{x\sim \setX}[x_{i^*}=j^*] \leq 1/m^{\tau}\le 1/2$ since $m\ge 4$ and $\tau=1/2$. 
    Consequently, 
    $\Pr_{x\sim\setX}[x_{i^*}\neq j^*] \geq 1/2$. 
    So, by \cref{fact:conditioningMinEntropy}, $\mintropy(\setX \mid_{x_{i^*}\neq j^*}) \geq \mintropy(\setX) -1$. 
    Therefore, the change in $\dinfy(\setX_{[N]\setminus I})$ is at most 1. This step is executed at most
    $B$ times.  (Note: For larger $m$ we could
    maintain sharper bounds, but it seems that we don't need to do so.)   
\item \autoref{step:chooseBiggerX} in Algorithm~\ref{algo:simulationAlgo}: choosing the more frequent bit of Alice to send:\\
This increases $\dinfy(\setX_{[N]\setminus I})$ by at most 1.   This step is executed $A$ times.
\item \autoref{step:updateXrestoredensity} and \autoref{step:updateRhoDueToXQuery} in Algorithm~\ref{alg:restoreX}: querying and fixing
coordinates $I'$ maximal for min-entropy loss:\\
First, \autoref{step:updateXrestoredensity} 
increases $\dinfy(\setX_{[N]\setminus I})$ by at most $\tau \cdot |I'| \cdot \log_2 m$ as shown
the proof of Lemma 3.5 in \cite{DBLP:journals/siamcomp/GoosPW20}.  
Second, \autoref{step:updateRhoDueToXQuery} decreases
$\dinfy(\setX_{[N]\setminus I})$ by precisely $|I'|\cdot \log_2 m$ since it adds
 $|I'|$ blocks to $I$.
 The net total of these changes is that 
 $\dinfy(\setX_{[N]\setminus I})$ decreases by
 at least $(1-\tau) \cdot |I'| \cdot \log_2 m$ in 
 this case.
\end{itemize}

Putting these together yields:   
$$\dinfy(\setX_{[N]\setminus I})\le A+B - (1-\tau)\cdot |I|\cdot \log_2 m.$$
Since $\dinfy(\setX_{[N]\setminus I})\ge 0$ we must
have $$|I|\le \frac{A+B}{(1-\tau)\log_2 m}.$$
Since $\tau=1/2$, $A+B\ge 0.5 |I|\log_2 m$ and hence
$C^{*,\oplus}(R\circ \indgadget_m^N)\ge A+B\ge 0.5\  C^{dt}(R) \log_2 m$.
\qed

We can strengthen the above in the case that each of Alice's
bits, like Bob's bits, either is irrelevant or splits
$\setX$ exactly in half.  
Then, as with our simulation of Bob's bits, we choose
to follow the side with the smaller protocol subtree.
We see that the paths followed in the
protocol $\detp$ are of total length (in bits that matter to the simulation) at most the
logarithm of the size of $\detp$.

In particular, this applies if $m$ is a power of 2 so that each $x_i$ is represented by a series of bits, 
Alice's bits are 
also parities, and we replace \autoref{step:updateXAvoidDependent} by constraining
one bit of $x_{i^*}$ that isn't already constrained to be different from the corresponding bit of $j^*$.
We write $L^{\oplus,\oplus}(R)$ for the number of
leaves (i.e., the size) of a protocol $\detp$ for $R$ in which
both Alice and Bob only send parities, which we call
\emph{parity communication}.
Using this 
obtain the following:

\begin{theorem}
\label{thm:fullparitylift}
$C^{dt}(R)$ is 
$O(\log_m(L^{\oplus,\oplus}(R\circ\indgadget_m^N)))$
\end{theorem}

\subsection{Parity decision trees and \texorpdfstring{Res$(\oplus)$}{Res(+)} proofs}





We can use the same ideas with small modifications to give a generic method for producing lower bounds for parity decision
trees from those for ordinary decision trees.

\begin{theorem}
\label{thm:pdt-lift}
For any sufficiently large $m$ that is a power of 2 and any function $f:\{0,1\}^N\rightarrow \set{0,1}$,
$\height^{dt}_\oplus(f\circ\indgadget_m^N)\ge\height^{dt}(f)$
and $\size^{dt}_\oplus(f\circ\indgadget_m^N)\ge 2^{\height^{dt}(f)}\ge \size^{dt}(f)$.
\end{theorem}

\begin{proof}
We follow the ideas of the proof of \cref{thm:fullparitylift} with a small modification that
combines the steps for Bob and for Alice as follows:
We maintain $\setX$ as an affine subspace as before and
we maintain a set $\dependent_\setY\subset [N]\times [m]$ of coordinates as before (though $\setY$ itself is not maintained), but now the equations that we maintain involving these coordinates
depend on the bits of the $x_i$ also.
At each parity that is not already implied, we row-reduce to remove the variables in $\dependent_\setY$.

If a coordinate $(i^*,j^*)$ in $\setY$ remains, we choose it as
the new dependent coordinate and complete the row-reduction.   We then apply the portion of the
simulation designated as Bob's simulation in order to
ensure that $x_{i^*}$ does not point to $j^*$.
If no such coordinate remains, then this becomes a
constraint on $\setX$ and we apply the portion associated with Alice.

In either case, we add one defining equation for $\setX$ for each bit that is sent that is not already
implied, so the deficiency of $\setX$ is precisely this
number.  
The same analysis shows that
$\height^{dt}(f)$ is 
$O(\log_m(\size^{dt}_\oplus(f\circ\indgadget_m^N)))$
and hence $O(\height^{dt}_\oplus(f\circ\indgadget_m^N)/\log m)$.
By choosing $m$ a sufficiently large constant, we obtain
that $\height^{dt}(f)\le \log_2(\size^{dt}_\oplus(f\circ\indgadget_m^N))$ and hence
$\size^{dt}(f)\le 2^{\height^{dt}(f)}\le \size^{dt}_\oplus(f\circ\indgadget_m^N)$ and
$\height^{dt}(f)\le \height^{dt}_\oplus(f\circ\indgadget_m^N)$.
\end{proof}

We can use this (in the form for relations, which we could easily have stated above) to show that we can convert each $k$-CNF $\varphi$ to an $O(k)$-CNF $\varphi\circ \indgadget_m^N$ that requires tree-like $Res(\oplus)$ refutations for $\varphi\circ \indgadget_m^N$ that are at least linear in the size of tree-like resolution refutations for $\varphi$.  

\begin{corollary}
For any sufficiently large integer $\ell$, $m=2^\ell$,
and any unsatisfiable $k$-CNF formula with $M$ clauses on $N$ Boolean variables, 
$\varphi\circ \indgadget_m^N$ is a $k(\ell+1)$-CNF formula with $M'=m^k M$ clauses on
$N'=N(\ell+m)$ variables that requires
$\treesize_{Res(\oplus)}(\varphi\circ\indgadget_m^N)\ge\treesize_{Res}(\varphi)\ge 2^{\width_{Res}(\varphi)-k}$.
\end{corollary}

\begin{proof}
By \cref{prop:reslin2}, for any minimal tree-$Res(\oplus)$ refutation of
$\varphi\circ \indgadget_m^N$ there is an isomorphic parity decision tree solving the search problem
$\search_{\varphi\circ \indgadget_m^N}$.
Using the form of \cref{thm:pdt-lift} for relations we can
convert a parity communication protocol for $\search_{\varphi\circ \indgadget_m^N}$ into one of at most the same size that solves the
search problem $\search_\varphi$, since each correct output of $\search_{\varphi\circ \indgadget_m^N}$ yields the name of a violated clause
of $\varphi\circ \indgadget_m^N$ that corresponds to a unique clause of
$\varphi$ and can be output by the ordinary decision tree.   
The final result follows using the equivalence of decision trees 
for $\search_\varphi$ and tree-resolution refutations of $\varphi$ from \cref{prop:dpll} and the tree-size/width 
relationship from \cref{prop:bsw}.
\end{proof}

We note that Itsykson and Kojevnikov~\cite{DBLP:conf/mfcs/ItsyksonS14,DBLP:journals/apal/ItsyksonS20}
previously used a much more specialized lifting theorem from~\cite{bps:kfoldtseitin-journal} for the specific
case of Tseitin formulas to give tree-like $Res(\oplus)$ lower bounds.  Our
new simple method is much more general and yields a large class of hard formulas.
We also note that Huynh and Nordstr\"{o}m~\cite{DBLP:conf/stoc/HuynhN12} gave lifting theorems for a variety of other proof systems using constant-size index gadgets (indeed with $m=3$) but these only yield good bounds for a restricted class of formulas whose search
problems have high ``critical block sensitivity".

\section{Summary and future directions}

Our results show that the Index (or Inner-Product) gadget is not a good disperser for low min-entropy deficiency rectangles $\setX \times \setY$ when $m$ is much smaller than $\log N$. 
Thus to reduce the gadget size beyond logarithmic using current techniques we need to consider other properties on the rectangle $\setX \times \setY$ maintained during the simulation  that can ensure that Index is a good disperser. 
Our counterexample to the conjecture of Lovett et al.~\cite{DBLP:conf/innovations/LovettMMPZ22} suggests the following natural property for $\setX$, $\setY$ in addition to having low entropy deficiency:
\begin{itemize}
    \item 
Except for a small subset of blocks $J$ of size $O(\Delta)$, every block of every $y$ in $\setY$ is “almost” balanced in terms of the number of zeroes and ones. That is, for any $i \in [N]\setminus J$ and for any $y \in \setY$, $ |y^i|_1 \approx m-|y^i|_1$, where $|y^i|_1$ denotes the number of ones in the $i$-th block of  $y$. 
\end{itemize}
Note that our counterexample is avoided by $\setY$ satisfying this property; 
indeed this property is violated in the extreme by our counterexample as every $y$ in the set $\setY$ we construct has $\omega(\Delta)$ blocks that are maximally unbalanced.  

The standard simulation paradigm allows considerable flexibility in choosing which subset to focus on in the rectangle of inputs associated with each node
of the communication protocol.   Maintaining something like the property above is easy to do and it is plausible that this or related properties will indeed be sufficient to yield general lifting theorems with very small Index gadgets.

Although our counterexample ruled out improving the gadget size to constant for Index in deterministic lifting theorems using current techniques, we were able to prove a lifting theorem  with constant-sized gadgets for the restricted class of protocols where Bob is restricted to sending parities. A natural extension of this direction is to consider semi-structured protocols where Bob is restricted to sending other interesting functions of his input bits. 
A natural class of restricted functions are threshold functions. 
The ideas used in our lifting theorem do not work in this case. 
As illustrated by the application of our semi-structured lifting theorem to Res$(\oplus)$ lower bounds, such restricted lifting theorems may have immediate applications in proof complexity. 
They may also be a natural avenue for developing new tools and techniques that could potentially help in proving general lifting theorems with constant sized gadgets.

\appendix

\bibliographystyle{plain}
\bibliography{references}

\end{document}